\documentclass[12pt]{article}
\usepackage{amssymb,amsmath,amsthm,latexsym,graphicx, bm, mathrsfs, hyperref}

\usepackage{amsmath,amssymb,amsfonts,latexsym}

\begin{document}
%<<<<<<<<<<< ennumeration of eqns section wise>>>>>>>>>>>>>>>>>>>

\renewcommand\theequation{\arabic{section}.\arabic{equation}}
\catcode`@=11 \@addtoreset{equation}{section}
%<<<<<<<<<<<<<<<<<<<<<<<<<<<<<<<<<>>>>>>>>>>>>>>>>>>>>>>>>>>>>>>>>>
\newtheorem{axiom}{Definition}[section]
\newtheorem{theorem}{Theorem}[section]
\newtheorem{axiom2}{Example}[section]
\newtheorem{claim}{Claim}[section]
\newtheorem{lem}{Lemma}[section]
\newtheorem{prop}{Proposition}[section]
\newtheorem{cor}{Corollary}[section]
\newcommand{\be}{\begin{equation}}
\newcommand{\ee}{\end{equation}}

\newcommand{\equal}{\!\!\!&=&\!\!\!}
\newcommand{\lmat}{\left(\begin{array}{cccccc}}
\newcommand{\rmat}{\end{array}\right)}
%Page length commands go here in the preamble
\setlength{\oddsidemargin}{-0.10in} % Left margin of 1 in + 0 in = 1 in
\setlength{\textwidth}{6.65in}   % Right margin of 8.5 in - 1 in - 6.5 in = 1 in
\setlength{\topmargin}{-.50in}  % Top margin of 2 in -0.75 in = 1 in
\setlength{\textheight}{8.8in}  % Lower margin of 11 in - 9 in - 1 in = 1 in

%\begin{document}

\title{Jacobi-Maupertius metric and Kepler equation \vspace{-0.1cm} \author{ Sumanto Chanda$^1$, G.W. Gibbons$^2$, Partha Guha$^1$ }}

\maketitle
\thispagestyle{empty}

\vspace{-0.55cm}

\begin{minipage}{0.46\textwidth}
\begin{flushleft}
\textit{\small $ ^1$ S.N. Bose National Centre \\ for Basic Sciences} \\
\textit{\small JD Block, Sector-3, Salt Lake, \\ Calcutta-700098, INDIA.} \\
\texttt{\small sumanto12@bose.res.in \\  partha@bose.res.in}
\end{flushleft}
\end{minipage}
\begin{minipage}{0.4\textwidth}
\begin{flushright}
\textit{\small $ ^2$  D.A.M.T.P., \\ University of  Cambridge} \\
\textit{\small Wilberforce Road, \\ Cambridge CB3 0WA,  U.K.} \\
\texttt{\small G.W.Gibbons@damtp.cam.ac.uk}
\end{flushright}
\end{minipage}

\smallskip

\abstract{This article studies the application of the Jacobi-Eisenhart lift, Jacobi metric and Maupertius transformation to 
the Kepler system. 
We start by reviewing fundamentals and the Jacobi metric. Then we study various ways to apply the lift to  Kepler related systems:
first as conformal description and Bohlin transformation of 
Hooke's oscillator, second in contact geometry,  third in Houri's transformation \cite{th}, coupled with  Milnor's 
construction \cite{Milnor}} with eccentric anomaly.  \\

{\bf MSC classes:} 70H06, 53D25, 53B20, 70F16. \\


\tableofcontents

\setcounter{page}{1}

\numberwithin{equation}{section}

\section{Introduction}

The Kepler system, derived by Johannes Kepler in 1609, as interpreted by Newton is a 3-dimensional integrable 
system for an inverse square law force describing elliptic trajectories \cite{Deriglazov, LM}. It is related to 
the oscillator system via a canonical transformation known as the Bohlin transformation, 
resulting in many properties of the two systems being inter-related. It has many integrals of motion such 
as the angular momentum, the Hamiltonian and the Runge-Lenz vector. The last two 
translate into the equivalent conserved quantities known as the Fradkin tensors for the 
oscillator system under Bohlin's transformation. Recently Kepler problem has been studied on noncommutative 
$\kappa$ -spacetime and corresponding Bohlin-Arnold duality \cite{GHZ}. In particular, regularization of the Kepler problem on
$\kappa$ -spacetime in several different ways \cite{GHZ1}.  Regularization is a mathematical procedure to cure this singularity. 
A nice clear treatment of regularizing the
Kepler problem was done by Moser in his 1970 paper \cite{Moser}, the treatment of Moser relates the Kepler
flow for a fixed negative energy level to the geodesic flow on the sphere $S^n$. A lucid analysis of
the geometrical aspects of Kepler problem can be found in Milnor \cite{Milnor}. Belbruno extended 
the cases of posetive energy to negative energy, in correspondence to the 3-hyperboloid 
$\mathcal{H}^3$, and zero energy which corresponds to 3-dimensional Euclidean space \cite{belbruno}.

\smallskip

The Jacobi-Maupertius (JM) metric is a projection of the  an action functional onto a fixed energy surface, 
reducing the problem to a spatial geodesic \cite{br}. In other words, the Jacobi-Maupertuis 
metric reformulates Newton's equations as geodesic equations for a Riemannian 
metric which degenerate at the Hill  boundary \cite{Mont}. An important application to gravity 
was shown \cite{Ong} by Ong who studied the curvature of the the Jacobi metric for the Newtonian 
$n$-body problem. For $n = 2$, the problem reduces to the Kepler's problem of the relative 
motion and the relevant Jacobi metric is up to an unimportant overall constant factor.
Recently, one of us  \cite{Gibbons} showed that free motion of massive particles in static spacetimes 
is given by geodesics of an energy-dependent Riemannian metric on the spatial sections 
analogous to Jacobi's metric in classical dynamics. Recently this result has been extended \cite{CGG} to explore the Jacobi metrics for 
various stationary metrics. In particular, the Jacobi-Maupertuis metric is formulated for time-dependent metrics by 
including the Eisenhart-Duval lift, known as the Jacobi-Eisenhart metric. 

\smallskip

This results in geodesic trajectory reparameterization, redefining the Hamiltonian 
and effectively making it a canonical transformation of the extended phase space to a 
conformal theory. All other conserved quantities of a system are preserved under this 
lift. The Bohlin transformation is possibly itself a Jacobi-Maupertius lift of the oscillator 
metric.  The JM metric plays an important role in statistical mechanics \cite{Biesiada,Pettini}. Krylov \cite{Krylov} 
suggests that viewing $n$-body dynamics as a geodesic flow on an appropriate manifold 
may provide a universal tool for discussing relaxation processes. \\

An $n$-dimensional system is Liouville integrable if it admits   $n$ first-
integrals in involution which the Lagrangian submanifold depends upon. This means that 
integrability is a geometric property that is independent of choice of parameterisation, as 
seen when the same conserved quantities remain unchanged under reparameterisation. 
The only consequence of a different choice of parameterisation would be to produce a new 
integrable models. \cite{CrM} have claimed that the harmonic oscillator, when it is reformulated in 
terms of JM geodesics, has positive Lyapunov exponents. \\

In this article, we shall first introduce some preliminaries and the basic formulation 
of the Jacobi metric \cite{th, avt} and show that conserved quantities are preserved and the 
equation of motion reduces to the geodesic equation \cite{br} under reparameterisation by 
examining the metric from the Lagrangian and Hamiltonian perspectives. Following 
that, the Kepler system will be shown to be geodesic flow on constant curvature surfaces. 
Here, we shall demonstrate how such a projection to a fixed energy surface following a 
canonical transformation is the Bohlin's transformation \cite{mls} that converts the oscillator 
system into the Kepler system. This will  be followed by a discussion on application in 
Houri's canonical transformation \cite{th}. First we shall couple it with Milnor's construction 
to study the preservation of the form of geodesic flows under such canonical transformations.

\smallskip

Organization of the paper is as follows. Section 2 outlines the basic aspects of the Jacobi-Maupertuis metric, 
Maupertuis transformation and integrable metrics. We also describe Jacobi-Eisenhart lift and Lagrange-Hamiltonian formulation 
in this section. We apply all these toolkits to Kepler equation in section 3.

\section{Preliminaries}

The Maupertuis form of the action defines the differential form of the action along the 
geodesic. From this form we are able to formulate the Hamilton's equations of motion, 
as well as the canonical transformations that are possible.
\begin{equation}
\label{action} S = \int_1^2 d \tau \ \mathcal{L} (x^\mu, \dot{x}^\mu) = \int_1^2 dt \ \big( p_i \dot{x}^i - H (x, p) \big)
\,.\end{equation}
Here, we deal with the extended phase-space which treats time as another co-ordinate 
$q^{n+1} = t$ and the Hamiltonian as its conjugate momentum $p_{n+1} = - H(x, p)$. According 
to Maupertuis, the dynamical path solution from the extremal of the action $S$ coincides 
with that of the reduced action $S_0$ for a fixed energy $H (x, p) = E$, given by:
\begin{equation}
\label{spaceact} S_0 = \int_1^2 dt \ p_i \dot{x}^i = \int_1^2 dt \ \frac{\partial \mathcal{L}}{\partial \dot{x}^i} \dot{x}^i
\,.\end{equation}
This reduced action is independent of any time evolution parameter, resulting in loss 
of information since we cannot restore the Hamiltonian function. The Jacobi-Eisenhart 
lift is one such process for dimensionally reducing geodesics. Such trajectories can be 
seen as geodesics of a corresponding configuration space or its enlargement under some 
constraints. Upon parametrizing as $\tau = t$, the time quadratic action term provides 
the potential. Since such Hamiltonians arise from Lagrangians with a metric origin, the 
Jacobi-Eisenhart lift reduces the dimensions of the geodesic. \\ \\
Given a metric on $n+1$ dimensional space-time $ds^2 = g^{n+1}_{\mu \nu} (x) \ dx^\mu dx^\nu$, it is simple 
to formulate the Lagrangian describing dynamics on the $n$ dimensional sub-space with a 
potential $U (x)$. Let $M$ be a manifold with local co-ordinates $x = \big( x^i \big), i = 1, . . . n$, with 
$x(t) \in M \subseteq \mathbb{R}^n, t \in [0, T]$ being a curve. Define the velocity as $\dot{x}(t) \in T_x M \subseteq \mathbb{R}^n$ and 
the momenta as $p(t) \in T_x^* M \subseteq \mathbb{R}^n$, $T_x M$, $T_x^* M$ being tangent and co-tangent spaces 
respectively at $x = x(t)$. 

If $\mathcal{L} = T - U; T_x M \longrightarrow \mathbb{R}$ is a natural Lagrangian, then $T$ is a non-degenerate 
quadratic form and $V$ is constant on each $T_x M$. Such dynamical systems under affine 
parametrization $\tau = x^0 = t$ are defined by the Lagrangian:
\begin{equation}
\label{lag} \mathcal{L} (x, \dot{x}) = \frac m2 g_{\mu \nu} (x) \dot{x}^\mu \dot{x}^\nu \equiv \frac m2 g_{ij} (x) \dot{x}^i \dot{x}^j - U (x) 
\,,\end{equation}
where $g_{\mu \nu} (x)$ is a Riemannian metric. The Euler-Lagrange equation is given by:
\begin{equation}
\label{accel} \ddot{x}^i = - \sum_{jk} \Gamma^i_{jk} \dot{x}^j \dot{x}^k - \sum_l g^{il} (x) \partial_l U(x)
\,.\end{equation}
The time-independent Hamiltonian $H: T^* M \longmapsto \mathbb{R}$ is a conserved quantity given by 
a Legendre transformation that maps the dynamics from the tangent to the co-tangent 
space $\mathbb{F} L: TM \longrightarrow T^*M; (x, \dot{x}) \longmapsto (x, p) = \bigg( x, \dfrac{\partial L}{\partial \dot{x}} \bigg)$.
\begin{equation}
  H (x, p) = \sum_{i=1}^n p_i \dot{x}^i - \mathcal{L} (x,
  \dot{x})\,, \qquad \qquad p_i = \frac{\partial \mathcal{L}}{\partial x^i} = g_{ij} (x) \dot{x}^j
\,.\end{equation}
If the Lagrangian has  a natural form given by (\ref{lag}),
then will  the Hamiltonian. 
The natural Hamiltonian for an autonomous system is a
conserved quantity. As seen in 
(\ref{action}) it acts as the generator for time-evolution of the geodesic action given by:
\begin{equation}
\label{potham} H (x, p) = \frac1{2m} g^{ij} (x) p_i p_j + U (x) \equiv T + U = E \,,
\end{equation}
where the Hamilton's dynamical equations are elaborated as:
\begin{equation}
\label{hamdyn} \dot{x}^i = \frac{\partial H}{\partial p_i} = \frac{g^{ij} (x)}m p_j \qquad \qquad \dot{p}_i = \frac{\partial H}{\partial x^i} = \frac1{2m} \frac{\partial g^{ij} (x)}{\partial x^i} p_i p_j + \frac{\partial U}{\partial x^i} \,.
\end{equation}
We have so far dealt with the action on a space-time manifold for a
general system 
with no fixed value for the Hamiltonian as a function on the
cotangent bundle. For an 
autonomous (time-independent) system,
we will have a fixed energy level defining the 
hypersurface on which motion takes place.
We will study the reduced action on this 
hypersurface, by employing the Jacobi metric arising from the
Jacobi-Maupertius lift.

\numberwithin{equation}{subsection}

\subsection{Jacobi-Maupertuis metric and Maupertuis principle}

It is possible to derive a metric which is given by the kinetic energy itself. Let us consider 
a conservative system with $n$ degrees of freedom whose Lagrangian is given by (\ref{lag}). The 
kinetic energy $T$ is a homogeneous function of degree $2$, hence Euler theorem implies 
$ 2T = \dot{x}^i \dfrac{\partial {\cal L}}{\partial \dot{x}^i}$, thus Maupertuis principle becomes
$$
\delta S = \delta\int_{t_1}^{t_2}2T\,dt = 0, \qquad \hbox{ where } \qquad T = \frac m2 g_{ij}(x) \dot{x}^i \dot{x}^j.
$$ 
Since the total energy of the conservative system is constant $E = T + U$, then substituting 
$U$ in the Lagrangian we find
${\cal L} = 2T-E$. Substituting in (\ref{action}) we obtain
\begin{equation}\label{mau}
\delta \int_{t_1}^{t_2}(2T -  E)\,dt = \delta \int_{t_1}^{t_2}2T\,dt - \delta \int_{t_1}^{t_2}E\,dt =  \delta \int_{t_1}^{t_2}2T\,dt.
\end{equation}
If we take the kinetic energy to be diagonal and all masses are equal, i.e., $a_{ij} = \delta_{ij}$, then 
equation (\ref{mau}) can be re-written as 
$$\delta \int_{t_1}^{t_2} 2T \ dt = \delta \int_{t_1}^{t_2} ds = \delta \int_{t_1}^{t_2} \big( \widetilde{g}_{ij}(x) \dot{x}^i\dot{x}^j \big)^{1/2} dt.$$
so that natural motion is geodesic of a configuration space $M$ and $ds$ is the differential 
arc length. The metric on
$M$ is referred to as the Jacobi metric, given by
\begin{equation}\label{e1}
ds^2 = \widetilde{g}_{ij}(x)dx^idx^j = 4\big(E - U(x)\big)T dt^2.
\end{equation}
Alternatively this expression can be obtained straight away by squaring $ds = 2T\, dt = 2(E - U(x))dt$. 
If we substitute $2 T = m g_{ij}(x)\dot{x}^i \dot{x}^j$ into (\ref{e1}) we find
$$
ds^2 = 2m \big(E - U(x)\big)g_{ij}(x)\dot{x}^i \dot{x}^j dt^2 = 2m \big(E - U(x)\big)g_{ij}(x)dx^idx^j.
$$
Thus we obtain
\begin{equation}\label{e2}
ds^2 = \widetilde{g}_{ij}(x) dx^idx^j, \qquad \hbox{ where } \qquad    \widetilde{g}_{ij}(x) = 2m (E - U(x)) g_{ij}(x).
\end{equation}
Thus the above arguments imply that a physical path of energy value $E$
is a geodesic with respect to the metric (\ref{e2}). So we can define the JM-metric
$\widetilde{g}_{jm}$ corresponding to an energy value $E$ of simple mechanical system $(M,g,U)$ as
$$
\label{jacg} \widetilde{g}_{jm} := 2m (E - U) g(x).
$$
This is the metric that defines the geodesic on the hypersurface of energy $E$ within spacetime. 
As we can see, the potential $U(x)$ has been merged into the metric components, 
giving the appearance of a potential-free space for a free article.

\subsection{Jacobi-Maupertuis transform and integrable metrics}

Let $M$ be a compact smooth $n$-dimensional Riemannian manifold with metric $g_{ij}(x)$. 
The cotangent bundle $T^{\ast}M$ is a smooth symplectic manifold with standard 2-form 
$\omega = \sum_{i=1}^{n}dp_i \wedge dx^i$. One can find other integrable metrics using the Jacobi-Maupertuis 
transformation. Let us consider the natural mechanical systems with Hamiltonian given 
by $H = \frac1{2m} \sum_{i,j}g^{ij}(x)p_ip_j + U(x)$. It is said that a Hamiltonian system on a $2n$-dimensional symplectic manifold is Louiville integrable if the $n$ first integrals are in involution and functionally independent everywhere. Integrable geodesic flows play a very important role in geometry, mechanics and integrable systems \cite{bt, BKF}. \\ \\
By the Maupertuis principle, for sufficiently large energy $E$, greater than $\hbox{ max } U({\bm x})$, 
on a fixed $(2n-1)$-dimensional smooth level surface $H(x, p) = E$, the integral trajectories 
of the vector field $X_H$ coincide with the trajectories of the another vector field ${\widetilde X_{\widetilde H}}$
corresponding to a new Hamiltonian ${\widetilde H}$ given by the formula
\begin{equation}\label{H1}
{\widetilde H}(x,p) = \frac1{2m} \sum_{i,j = 1}^{n}\frac{g^{ij}(x)}{E-U(x)} p_i p_j,
\end{equation}
The Maupertuis transformation $X_H \to {\widetilde X_{\widetilde H}}$ relates two vector field on $M$. If $t$ and $\sigma$ 
are time along trajectories of the vector fields $X_H$ and ${\widetilde X_{\widetilde H}}$, then 
\begin{equation}
\label{t1} d\sigma = (E - U(x)) dt
\end{equation}
The distinguished role of the time $t$ is not desirable in the
general case of non-autonomous Hamiltonian systems. We therefore
introduce an evolution parameter $s$ to parameterize
 time evolution of the system. In the extended formalism, time $t$ is
treated as an ordinary canonical function $t(s) \equiv x^{0}(s)$
of an evolution parameter $s$. We may conceive a `new' momentum
coordinate $p_0 (s)$ in conjunction with the time as an additional
pair of canonically conjugate coordinates. The extended
Hamiltonian ${\cal H}(x^0,p_0,x^i,p_i)$ is defined as a
differentiable function on the cotangent bundle $T^{\ast}Q =
T^{\ast}({\Bbb R} \times M)$ endowed with a chart $(p_0, p_i) \in
T_{x_0,x_i}^{\ast}Q$ with $\frac{\partial {\cal H}}{\partial s} =
0$. It is given by ${\cal H}(x^0,p_0,x^i,p_i) = H(x^i,p_i,x^0) +
p_0$, where $x^0$ and $p_0$ are conjugate variables and $p_0 = -H$. The extended phase space admits a
Liouville form (or integral invariant of Poincar\'e-Cartan) 
\begin{equation}
\label{eq} {\cal \theta}_{\cal H} = p_0dt + p_idx^i 
\end{equation}
 and the Hamiltonian flow is completely determined by
the conditions:
$$
\langle{\Bbb X}_{\cal H}, dt \rangle = 1\;\;\; \mbox{and} \qquad {\Bbb
X}_{\cal H} \lrcorner d{\cal \theta}_{\cal H} = 0,
$$
where 
\begin{equation}
{\Bbb X}_{\cal H} = \dot{x}^\mu \frac{\partial \ }{\partial x^\mu} + \dot{p}_\mu \frac{\partial \ }{\partial p_\mu}
\end{equation}
Invoking Hamilton's equations of motion, and keeping in mind that $\dot{t} = 1, p_0 = - q(t)$ 
and the Maupertuis form of action, we have the extended Hamiltonian given below
$$\text{Maupertuis form of action: } \ \mathcal{L} (x^\mu, \dot{x}^\mu) = \sum_{\mu = 0}^n p_\mu \dot{x}^\mu = \sum_{i = 1}^n p_i \dot{x}^i + p_0 \dot{t}$$
$$\mathcal{H} (x^i, p_i, t) = \sum_{\mu = 0}^n p_\mu \dot{x}^\mu - \mathcal{L} (x^\mu, \dot{x}^\mu) = \Big[ \sum_{i = 1}^n p_i \dot{x}^i - \mathcal{L} (x^i, \dot{x}^i, t) \Big] + p_0 \dot{t} = 0$$
$$\mathcal{H} (x^i, p_i, t) = H (x^i, p_i, t) - q(t) = 0$$
Thus, the extended Hamiltonian vector field is given by
$${\Bbb X}_{\cal H} = \sum_\mu \bigg( \frac{\partial \mathcal{H}}{\partial p_\mu} \frac{\partial \ }{\partial x^\mu} - \frac{\partial \mathcal{H}}{\partial x^\mu} \frac{\partial \ }{\partial p_\mu} \bigg) \quad = \quad \sum_i \bigg( \frac{\partial \mathcal{H}}{\partial p_i} \frac{\partial \ }{\partial x^i} - \frac{\partial \mathcal{H}}{\partial x^i} \frac{\partial \ }{\partial p_i} \bigg) + \frac{\partial \mathcal{H}}{\partial H} \frac{\partial \ }{\partial t} - \frac{\partial \mathcal{H}}{\partial t} \frac{\partial \ }{\partial H}$$
Here, we apply some rules:
$$\frac{\partial \mathcal{H}}{\partial x^i} = \frac{\partial H}{\partial x^i}, \qquad \frac{\partial \mathcal{H}}{\partial p_i} = \frac{\partial H}{\partial p_i}, \qquad \qquad \frac{\partial \mathcal{H}}{\partial H} = 1, \qquad \frac{\partial \mathcal{H}}{\partial t} = 0$$
 \begin{equation} 
{\Bbb X}_{\cal H} = \frac{\partial {H}}{\partial p_i}\frac{\partial}{\partial x^i} - \frac{\partial {H}}{\partial x^i}\frac{\partial}{\partial p_i} + \frac{\partial}{\partial t}
\end{equation}
is the time-dependent Hamiltonian vector field. The vector field ${\Bbb X}_{\cal H}$ lies in the kernel of $d{\cal \theta}_{\cal H}$,
so the bicharacteristic of ${\cal \theta}_{\cal H}$ is a path through the extended phase space such that the tangent vector
to the path at any point is parallel to ${\Bbb X}_{\cal H}$. \\ \\
It is clear that the Poincar\'e-Cartan two form associated to (\ref{eq})
\begin{equation}
{\cal \omega} = \sum_i dp_i \wedge dx^i - dH \wedge dt
\end{equation}
is invariant under Jacobi-Maupertuis transformation. This reveals that the JM transformation
is the {\it time-dependent canonical transformation}. \\ \\
Consider the time-dependent canonical transformations of the extended phase space,
\begin{equation}
\label{repara} 
\begin{split}
t \to \sigma &\qquad d\sigma = \Lambda (x, p) dt \\
H \to {\widetilde H} &\qquad {\widetilde H} = \Lambda^{-1} (x, p) H
\end{split}
\end{equation}
where $\Lambda (x, p) = \big(E - U (x) \big)$. This changes the initial equations of motion
$$
\frac{dx^i}{d\sigma} = \Lambda^{-1}(x, p)\bigg(\frac{dx^i}{dt} - {\widetilde H}\frac{\partial \Lambda}{\partial p_i} \bigg),
\qquad \frac{dp_i}{d\sigma} = \Lambda^{-1}(x, p)\bigg(\frac{dp_i}{dt} + {\widetilde H}\frac{\partial \Lambda}{\partial x^i}\bigg).
$$
This preserves the canonical form of the Hamilton-Jacobi equation given by
%\begin{equation}
%\frac{\partial S}{\partial t} = - H (x, p) \hspace{2cm} S = \int_1^2 \bigg( p_i d x^i - H (x, p) dt \bigg),
%\end{equation}
$$
\frac{\partial S}{\partial \sigma} + {\widetilde H} = \frac{\partial S}{\partial t}\frac{dt}{d\sigma} + \Lambda^{-1}H 
= \Lambda^{-1} \bigg( \frac{\partial S}{\partial t} + H \bigg) = 0.
$$
In other words, $S$ satisfies
$$
S = \int (p_i dx^i - Hdt) = \int (p_i dx^i - {\widetilde H}d\sigma).
$$
Integral trajectories have two parametric forms $X_H$ and $X_{\widetilde H}$ corresponding to the Hamiltonians $H$ and 
${\widetilde H} = \Lambda^{-1} (x, p) H$ respectively. The transformation $X_H \to X_{\widetilde H}$ is the
Maupertuis transformation. If $\sigma$ be the time along trajectories of the vector $X_{\widetilde H}$, then the
Maupertuis transformation gives the Jacobi transformation $d\sigma = (E - U (x))dt$. \\ \\
Thus, the reparameterization can be seen as part of the canonical transformation \cite{iat1, iat2} to counter the changes in the form of the equation of motion. This maps the geodesic onto another geodesic while preserving integrability.

\subsection{Jacobi-Eisenhart lift and Jacobi-Maupertuis metric}

The Jacobi-Eisenhart lift eliminates the potential in a $n$ dimensional Hamiltonian system (\ref{potham}), reducing it into a spatial $n$-dimensional free particle geodesic. The result is the Jacobi metric which, for time-independent Hamiltonian systems, projects the original geodesic onto a constant energy surface as a spatial geodesic describing a free particle with no potentials. \\ \\
So, for a fixed energy $H (x, p) = E$, the Jacobi Hamiltonian $\widetilde{H}$ \cite{th, avt} is:
\begin{equation}
\label{jacham} g^{ij} (x) p_i p_j = 2m \big[ E - U (x) \big] \qquad \Rightarrow \qquad \widetilde{H} = \frac1{2m} \frac{g^{ij} (x)  p_i p_j}{E - U (x)} \equiv \frac T{E - U (x)} = 1 
\end{equation}
This means that the metric and its inverse transform into their Jacobi-Maupertius equivalent as shown below:
\begin{equation} \label{jacmet}
\begin{split}
\widetilde{g}^{ij} (x) p_i p_j &= 1 \\ 
\widetilde{g}^{ij} (x) = \frac{g^{ij} (x)}{2m \big[E - U (x) \big]} \qquad &\Rightarrow \qquad \widetilde{g}_{ij} (x) = 2m \big[ E - U (x) \big] g_{ij} (x)
\end{split}
\end{equation}
Naturally, for a transformed Hamiltonian, the dynamical description should also 
change to match the new generator of time translations. This essentially means that the 
geodesic must be reparameterized to keep the form of Hamilton's equations invariant. 
Furthermore, from the lifted Hamiltonian, using (\ref{hamdyn}) and (\ref{jacham}) gives the momentum 
and reparameterization factor:
$$\frac{dx^i}{d\sigma} = \frac{\partial \widetilde{H}}{\partial p_i} = 2 \widetilde{g}^{ij} (x) p_j = \frac1{m \big[E - U (x) \big]} \Big( g^{ij} (x) p_j \Big) = \frac{dt}{d\sigma} \dot{x}^i$$
\begin{equation} 
\label{reparfac} p_i = \frac12 \widetilde{g}_{ij} (x) \frac{dx^j}{d\sigma} \qquad \qquad \Lambda (x, p) = \frac{d\sigma}{dt} = \big| E - U (x) \big|
\end{equation}
Thus, according to (\ref{repara}), the new Hamiltonian can be said to be:
\begin{equation}
\label{newham} \mathcal{H} = \frac{H}{\big| E - U (x) \big|}
\end{equation}
Using (\ref{reparfac}) for the Jacobi Hamltonian, we can say that the reduced Lagrangian is
\[ \begin{split}
\widetilde{\mathcal{L}} = \widetilde{g}_{ij} (x) \frac{dx^i}{d\sigma} \frac{d x^j}{d\sigma} = 4 \widetilde{g}^{ij} (x) \bigg(  \frac12 \widetilde{g}_{ik} (x) \frac{d x^k}{d\sigma} \bigg) \bigg( \frac12 \widetilde{g}_{jl} (x) \frac{dx^l}{d\sigma} \bigg) = 4 \widetilde{g}^{ij} (x) p_i p_j = 4 \widetilde{H} = 4
\end{split} \]

\begin{equation}
\therefore \qquad \widetilde{\mathcal{L}} = \big[ E - U (x) \big] \mathcal{L} \quad = \quad \widetilde{g}_{ij} (x) \frac{dx^i}{d \sigma} \frac{d x^j}{d \sigma} = 4
\end{equation}

Liouville integrability of an $n$-dimensional geodesic flow is defined to imply that: 

\begin{enumerate}

\item[a.] $n$ functionally independent first-integrals of motion $I_n$ exist almost everywhere.

\item[b.] Such integrals are in involution: $\{ I_j, I_k \} = 0$ for all $1 < j, k < n$.

\end{enumerate}
Restricting the geodesic flow onto any non-zero fixed energy level surfaces are smoothly 
equivalent to the trajectory. Consequently, we may redefine the condition of integrability 
to imply the existence of $n-1$ functionally independent first integrals in involution 
almost everywhere on the unit covector bundle $\{ \widetilde{H} (x, p) = \widetilde{g}^{ij} (x) p_i p_j = 1 \} \subset T^* M^n$ \cite{bt}.

\subsection{Lagrangian and Hamiltonian perspective}

Mechanics has been historically studied from two approaches: Lagrange's and Hamilton's. This results in two different, yet equivalent formulations of the equations of motion to describe geodesics. Since we have shown how to formulate the lifted Hamiltonian and Lagrangian, it is natural to explore how the equations of motion take shape under such formulations, and the effect on conserved quantities. \\ \\
Starting with the Hamiltonian in (\ref{jacham}), we shall write the dynamical equations with respect to a new parameter $\sigma$ as shown in \cite{br}
\begin{equation} \label{jacdyn}
\begin{split}
\frac{d x^i}{d \sigma} &= \frac{\partial \widetilde{H}}{\partial p_i} = \frac{g^{ij} (x)}{m \big[ E - U (x) \big]} p_j \\
\frac{d p_i}{d \sigma} &= - \frac{\partial \widetilde{H}}{\partial x^i} = - \frac1{E - U (x)} \bigg[ \frac1{2m} \frac{\partial g^{mn} (x)}{\partial x^i} p_m p_n + \frac{\partial U}{\partial x^i} \bigg]
\end{split}
\end{equation}

\begin{theorem}
Let $T: TQ \rightarrow \mathbb{R}$ be a smooth pseudo-Riemannian metric, $U : Q \rightarrow \mathbb{R}$ 
be a smooth potential energy function, and $t \mapsto q(t), I \rightarrow Q$ be a curve in $Q$ such that 
$E \big(\bm{q}(t), \bm{\dot{q}} (t) \big) = E \in \mathbb{R}$ and $U (\bm{q}(t)) \neq E \ \forall \ t \mapsto \widetilde{t}(t), I \rightarrow \mathbb{R}$ defined by
$$\sigma(t) = \int_0^t d \tau \big( E - U (q(\tau)) \big)$$
is a diffeomorphism into its image $J: s \mapsto t(s), J \rightarrow I$. Moreover, $t \rightarrow q(t)$ in $Q$ is a 
solution to the Euler-Lagrange equation $EL(L) = 0$ iff the curve $s \mapsto x(t(s)), J \rightarrow Q$ is 
a geodesic of the Jacobi metric $\widetilde{L} = (E - U) L$.
\end{theorem}

\begin{proof}
So long as we have
$$ \frac{d \sigma(t)}{dt \ } = E - U (x(t)) \neq 0$$
the inverse function theorem guarantees that $t \mapsto s(t)$ is a diffeomorphism onto its 
image $Q$, reparameterizing the the curve as $s \mapsto x(s) = x(t(s))$. Thus, the velocity upon 
differentiation wrt $t$ is:
\begin{equation}
\frac{d x^i}{d t} = \frac{d x^i}{d \sigma} \frac{d \sigma}{d t} = \big( E - U (x) \big) \frac{d x^i}{d \sigma}
\end{equation}
and the acceleration from (\ref{accel}) can be re-written as:
\begin{equation}
\ddot{x}^i = \frac{d \sigma}{dt} \frac{d \ }{d \sigma} \bigg[ \frac{d \sigma}{dt} \frac{d x^i}{d \sigma} \bigg] = \big( E - U (x) \big)^2 \frac{d^2 x^i}{d \sigma^2} - \big( E - U (x) \big) \partial_j U (x) \frac{d x^i}{d \sigma} \frac{d x^j}{d \sigma}
\end{equation}
and the Euler-Lagrange equation (\ref{accel}) transforms as:
$$\big( E - U (x) \big) \frac{d^2 x^i}{d \sigma^2} - \partial_j U (x) \frac{d x^i}{d \sigma} \frac{d x^j}{d \sigma} = - \sum_{jk} \big( E - U (x) \big) \Gamma^i_{jk} \frac{d x^j}{d \sigma} \frac{d x^k}{d \sigma} - \sum_l \widetilde{g}^{il} \partial_l U(x)$$
\[ \begin{split}
\Gamma^i_{jk} &= \frac12 g^{im} \big( \partial_j g_{mk} + \partial_k g_{mj} - \partial_m g_{jk} \big) \\
&= \bigg[ \frac 1{2 \big( E - U (x) \big)} \bigg( \partial_j U (x) \delta^i_k + \partial_k U (x) \delta^i_j - \widetilde{g}^{im} \partial_m U (x) \widetilde{g}_{jk} \bigg) + \widetilde{\Gamma}^i_{jk} \bigg]
\end{split} \]
\[ \begin{split}
\big( E &- U (x) \big) \frac{d^2 x^i}{d \sigma^2} - \partial_l U (x) \frac{d x^l}{d \sigma} \frac{d x^i}{d \sigma} = - \sum_{jkl} \bigg[ \big( E - U (x) \big) \Gamma^i_{jk} \frac{d x^j}{d \sigma} \frac{d x^k}{d \sigma} + \widetilde{g}^{il} \partial_l U(x) \bigg] \\
& =  - \sum_{jkl} \bigg[ \partial_l U (x) \frac{d x^l}{d \sigma} \frac{d x^i}{d \sigma} - \widetilde{g}^{im} \partial_m U (x) \bigg( \underbrace{\frac12 \widetilde{g}_{jk} \frac{d x^j}{d \sigma} \frac{d x^k}{d \sigma}}_{\frac12 \widetilde{\mathcal{L}}=2} \bigg) + \widetilde{g}^{il} \partial_l U(x) \bigg] \\
& \qquad \qquad - \sum_{jk} \big( E - U (x) \big) \widetilde{\Gamma}^i_{jk} \frac{d x^j}{d \sigma} \frac{d x^k}{d \sigma}
\end{split} \]
\begin{equation}
\label{lifteom} \frac{d^2 x^i}{d \sigma^2} = - \widetilde{\Gamma}^i_{jk} \frac{d x^j}{d \sigma} \frac{d x^k}{d \sigma}
\end{equation}
Thus, the Euler-Lagrange equation has been mapped to a regular geodesic equation for 
the Jacobi metric (\ref{jacg}). The Jacobi-Maupertius principle holds for any system with 
non-zero kinetic energy.
\end{proof}

Also, for any conserved quantity $K = K^{(2)ij} p_i p_j + K^{(0)}$, we can say:
\begin{align}
\frac{d K}{d \widetilde{t}} = \big\{ K, \widetilde{H} \big\} = \frac{d t}{d \widetilde{t}} \frac{d K}{d t} = \frac1{E - U (x)} \big\{ K, H \big\} \\ \nonumber \\
\therefore \qquad \big\{ K, \widetilde{H} \big\} = 0 \qquad \Rightarrow \qquad \big\{ K, H \big\} = 0
\end{align}
In \cite{th}, T. Houri describes $\widetilde{K} = K^{(2)ij} p_i p_j + K^{(0)} \widetilde{H}$ where according to (\ref{jacham}):
\begin{equation}
\widetilde{K} = K^{(2)ij} p_i p_j + K^{(0)} \widetilde{H} \hspace{0.25cm} = \hspace{0.25cm} K^{(2)ij} p_i p_j + K^{(0)} = K \hspace{2cm} \because \hspace{0.25cm} \widetilde{H} = 1
\end{equation}
Thus, showing that the conserved quantities remain the same for the Jacobi metric. This is not surprising given that the Jacobi -Eisenhart lift was just a reparameterization that left position and momenta unaltered. Since all conserved quantities or first integrals in Hamiltonian mechanics are polynomials of position and momenta, they should also be unchanged under such a transformation, unless a canonical transformation is involved. \\ \\
We shall now proceed to apply the the Jacobi metric to the Kepler problem.

\numberwithin{equation}{section}

\section{Application to Kepler problem}

We now consider the Kepler problem of orbital motion in the presence of a central 
potential $U(r) = - \frac{\alpha}r$. Since this is a problem involving spherical symmetry, we have the 
spatial part of the metric as the conformally flat polar metric. We shall only consider 
two dimensional motion because of angular momentum conservation in a radial potential. \\ \\
Thus, the Jacobi-Kepler metric is given as a conformally flat metric:
\begin{equation}
d \widetilde{s}^2 = \big( E - U(r) \big) \big( d r^2 + r^2 d \theta^2 \big) = f^2 (r) \big( d r^2 + r^2 d \theta^2 \big)
\end{equation}
Here, the Gaussian curvature is given by:
\[ \begin{split}
e^r = f(r) \ dr \qquad &\qquad e^\theta = r f(r) \ d \theta \\
d e^\theta = \big( r f(r) \big)' dr \wedge d \theta \qquad &\Rightarrow \qquad {\omega^\theta}_r = \frac{\big( r f(r) \big)'}{f(r)} d \theta \\ \\
d {\omega^\theta}_r = \bigg( \frac{\big( r f(r) \big)'}{f(r)} \bigg)' dr \wedge d \theta \qquad &\Rightarrow \qquad {R^\theta}_{r \theta r} = - \frac1{r f^2(r)} \bigg( \frac{\big( r f(r) \big)'}{f(r)} \bigg)'
\end{split} \]
\begin{equation}
\label{gacurv} \therefore \qquad K_G = {R^\theta}_{r \theta r} = - \frac1{r f^2 (r)} \frac{d \ }{dr} \bigg( \frac1{f(r)} \frac{d \ }{dr} \big( r f(r) \big) \bigg)
\end{equation}
Thus, for $f^2 (r) = E - U(r)$, the Gaussian curvature (\ref{gacurv}) in this case is given as:
\begin{equation} 
K_G = \frac{\big( r U' (r) \big)' \big( E - U(r) \big) + r \big( U'(r) \big)^2 }{2 r \big( E - U(r) \big)^3} 
\end{equation}
If $h$ is a regular value of $U(r)$ on the boundary ring, ie. $U(r) = h; x \in \partial M$ we have by continuity
\begin{equation}
\big( r U' (r) \big)' \big( E - U(r) \big) + r \big( U'(r) \big)^2 > 0, \qquad K_G \longrightarrow \infty
\end{equation}
In case of the Kepler problem, we have $U(r) = - \dfrac{\alpha}r$, so the Gaussian curve $K_G$ is:
\begin{equation}
\label{kcurv} K_G = - \frac{\alpha E}{2 \big( r E + \alpha \big)^3}.
\end{equation}
Thus, we can see that the curvature is classified as:
\begin{equation}
 \forall \quad E > - \frac{\alpha}r \qquad
\begin{cases}
E < 0 \quad \Rightarrow \quad K_G > 0 \quad ; \quad \text{ellipse} \\
E = 0 \quad \Rightarrow \quad K_G = 0 \quad ; \quad \text{parabola} \\
E > 0 \quad \Rightarrow \quad K_G < 0 \quad ; \quad \text{hyperbola} 
\end{cases}
\end{equation}
Thus, for the Kepler problem, for negative energies in the range $- \frac{\alpha}r < E < 0$, we will 
have positive curvature, and thus closed periodic orbits described by the Jacobi-Kepler 
metric. What motivates us to connect this theory with the Kepler problem is that it 
describes $\widetilde{H} = 1$ geodesic flow on $T^* S^3, K_G = 1$ energy surface. \\ \\
The Hamiltonian flow along a geodesic is given by the Hamiltonian vector field operator, 
which for the Kepler equation, essentially becomes:
\begin{equation}
\label{kepflow} X_{H} = \frac{\partial H}{\partial p_i} \frac{\partial \ }{\partial x^i} - \frac{\partial H}{\partial x^i} \frac{\partial \ }{\partial p_i} \hspace{0.5cm} = \hspace{0.5cm} p_i \frac{\partial \ }{\partial x^i} - \alpha \frac{x^i}{r^3} \frac{\partial \ }{\partial p_i}
\end{equation}
Thus, under circumstances of constant curvature, the radial equation of motion is:
\begin{equation}
\ddot{r} - \frac{p_\theta^2}{r^3} = - U'(r)
\end{equation}
Thus, for constant vanishing Gaussian curvature, we will have the Kepler potential, and 
thus, the Kepler equations of motion. However, if we consider the Jacobi metric and 
Hamiltonian, we will have:
\begin{align}
\frac{d r}{d \sigma} = \frac1{E - U (r)} &p_r \qquad \Rightarrow \qquad p_r = \frac{rE + \alpha}r \frac{d r}{d \sigma} \\
\frac{d p_r}{d \sigma} &= - \frac r{rE + \alpha} \bigg[ - \frac{p_\theta^2}{r^3} + \frac{\alpha}{r^2} \bigg] \\ \nonumber \\ \displaybreak[0]
- \frac{\alpha}{r^2} \bigg( \frac{dr}{d \sigma} \bigg)^2 &+ \frac{rE + \alpha}r \frac{d^2 r}{d \sigma^2} = - \frac r{rE + \alpha} \bigg[ - \frac{p_\theta^2}{r^3} + \frac{\alpha}{r^2} \bigg] \nonumber \\ \nonumber \\ 
\label{radeq} \therefore \qquad \frac{d^2 r}{d \sigma^2} &= - \frac{E p_\theta^2}{\big( rE + \alpha \big)^3} - \frac{\alpha}{\big( r E + \alpha \big)^2}
\end{align}
If one wishes to verify, it can be confirmed in (\ref{radeq}) that:
\begin{equation}
\widetilde{\Gamma}^r_{jk} \frac{d x^j}{d \sigma} \frac{d x^k}{d \sigma} = \frac{E p_\theta^2}{\big( rE + \alpha \big)^3} + \frac{\alpha}{\big( r E + \alpha \big)^2}
\end{equation}
showing that the RHS of (\ref{radeq}) matches that of (\ref{lifteom}), and our analysis is consistent.

\numberwithin{equation}{subsection}

\subsection{Bohlin transformation and duality}

The Bohlin transformation is a canonically converts the dynamics of the oscillator system into that of the Kepler system and vice versa. We shall see how the Jacobi metric for a fixed energy following a canonical transformation demonstrates this as shown in \cite{avt}. \\ \\
The transformation rule involves expressing the co-ordinates as a complex variable:
\begin{equation}
r = q_1 + i q_2
\end{equation}
The canonical transformation we shall use as shown in \cite{mls} is:
\begin{align}
r \quad \longrightarrow \quad z = \frac{r^2}2 = \bigg( \frac{q_1^2 - q_2^2}2 \bigg) &+ i \big( q_1 q_2 \big) = x + i y \qquad
\begin{cases}
x = \dfrac{q_1^2 - q_2^2}2 \\
y = q_1 q_2
\end{cases} \\
 x^2 + y^2 &= \frac{\big( q_1^2 + q_2^2 \big)^2}4, \quad \hbox{ or } \quad 2\sqrt{x^2 + y^2} = q_{1}^{2} + q_{2}^{2}.
\end{align}
For the covariant momentum, in accordance with Bohlin's transformation rule:
\begin{align}
 \begin{array}{c}
p_1 = \dfrac{\partial x}{\partial q_1} p_x + \dfrac{\partial y}{\partial q_1} p_y = \ q_1 p_x + q_2 p_y \\ \\
p_2 = \dfrac{\partial x}{\partial q_2} p_x + \dfrac{\partial y}{\partial q_2} p_y = -  q_2 p_x + q_1 p_y 
\end{array} \Bigg\}
& \qquad p = p_1 + i p_2 = \big( q_1 - i q_2 \big) \big( p_x + i p_y \big)
\end{align}
%\begin{align}
%\therefore \qquad p_1^2 + p_2^2 = \big( q_1^2 + q_2^2 \big) &\big( p_x^2 + p_y^2 \big) = 2 \sqrt{ x^2 + y^2 } \big( p_x^2 + p_y^2 \big)  \nonumber
%\end{align}
This transformation can also be written in matrix form as:
\begin{equation}
\left({\begin{array}{c}
p_x \\ p_y
\end{array} } \right) =
\frac1{q_1^2 + q_2^2}\left({\begin{array}{cc}
q_1 & - q_2 \\
q_2 & q_1
\end{array} } \right)
\left({\begin{array}{c}
p_1 \\ p_2
\end{array} } \right)
\end{equation}
Thus we obtain
\begin{align}
\qquad \frac{p_1^2 + p_2^2}{q_1^2 + q_2^2} =  p_x^2 + p_y^2. 
\end{align}
Let $H(p,q)$ be any Hamiltonian and fix the energy $E$. Let us consider flow by the 
reparametrization $\frac{dt}{d\tau} = f(q,p)$
This immediately yields
$$
{\widetilde H}(p,q) = f(p,q)(H(p,q) - E), $$
which retains the zero energy surface on the level set of $H$ to the energy $E$
$$
H^{-1}(E) = \{ (p,q) | H(p,q) = E \}.
$$
If the oscillator Hamiltonian is given as
\begin{equation}
H_{osc}( q_i, p_i ) = \frac12 \big( p_1^2 + p_2^2 \big) + \frac{a}2 \big( q_1^2 + q_2^2 \big) - b
\end{equation}
The transformation (3.14) maps the Hamiltonian of the oscillator equation to that of 
Kepler,
\begin{equation}\label{kep3.21}
H_{kepler}( x,p) = p_x^2 + p_y^2 - \frac{b}{2\sqrt{ x^2 + y^2 }} + a,
\end{equation}
thus (3.14) can be considered to be the Bohlin transformed co-ordinates and
for the time being we assume $ r = \sqrt{ x^2 + y^2 } \neq 0$. This clearly yields the transformation of the oscillator hamiltonian into the Kepler Hamiltonian. \\ \\
If $a$ and $b$ are treated as new momenta then the null lift of the (\ref{kep3.21}), given by
\begin{equation}\label{kep3.22}
{\widetilde H}( x,p) = p_x^2 + p_y^2 - \frac{pz^2}{\sqrt{ x^2 + y^2 }} + p_{a}^{2},
\end{equation} 
where we have added two new conjugate variables $(z,p_z)$, $(a,p_a)$ and corresponding 
momenta being conserved. Recently, Cariglia \cite{carig} made a fine observation to connect all the 
energy (positive, null and negative) regimes of Kepler orbit by introducing an additional 
conjugate pair. This one can be done if we replace $(a,p_a)$ pair by two additional conjugate 
pair $(\alpha, p_{\alpha})$ and $(\gamma,p_{\gamma})$ and Hamiltonian ${\widetilde H}( x,p)$ is replaced by
$${\widetilde {\cal H}}( x,p) = p_x^2 + p_y^2 - \frac{pz^2}{\sqrt{ x^2 + y^2 }} - p_{\alpha}^{2} + p_{\gamma}^{2}. $$
%\begin{equation}
%H ( x_a, p_a ) = \sqrt{ x^2 + y^2 } \bigg[ \big( p_x^2 + p_y^2 \big) + a^2 \bigg] + b
%\end{equation}
%Thus, since $V (q_1, q_2) = \dfrac{a^2}2 \big( q_1^2 + q_2^2 \big) + b = a^2 \sqrt{ x^2 + y^2 } + b$, if the Energy $E = b$, then the new Hamiltonian $\mathcal{H}$ according to (\ref{newham}) is given by:
%
%\begin{equation}
%\mathcal{H} ( x, y ) = \frac{H ( x, y )}{\big| E - U (x) \big|} = \frac{\big( p_x^2 + p_y^2 \big)}{a^2} + \frac b{2 a^2 \sqrt{ x^2 + y^2 }} + 1
%\end{equation}
%Thus, we have converted the oscillator hamiltonian into the Kepler Hamiltonian using Bohlin's canonical transformation followed by the computation of the Jacobi-Maupertius metric that projects the system onto a specific fixed energy surface.
%

\subsection{Contact method, reparametrization and regularization}

A contact form $\alpha$ on a $(2n+1)$-dimensional manifold $M$ is a Pfaffian form satisfying 
$\alpha \wedge (d\alpha)^n \neq 0$. The contact distribution is given by ${\cal C}|_{U} = \hbox{ Ker }\alpha |_{U}$, where $U$ is the open 
set in $M$. Given a contact form $\alpha$ , the Reeb vector field $Z$ is a vector field uniquely 
defined by   
\begin{equation}
i_Z\alpha =1, \qquad i_Z d\alpha = 0.
\end{equation}
Here we are interested in problem of closed Hamiltonian trajectories on a fixed energy 
$H = E$ surface, so we follow Weinstein's method. Let $P^{2n}$ be the total space of the 
principle ${\Bbb R}^{\ast}$-bundle $\pi : P \to M$, whose fibers are non-zero covectors $(q,p)$ that vanish 
on the contact element ${\cal C}(x)$ in $M$. The symplectization $P$ has a canonical $1$-form $\alpha$, 
restriction of Liouville $1$-form, and the symplectic form is given by $\omega = d\alpha$. Consider 
the multiplicative ${\Bbb R}^{\ast}$ action on $(P,\omega)$, from the nongeneracy of $\omega$, there exist a unique 
vector field $Y$, called the Liouville vector field, which satisfies the following identities:
\begin{equation} 
i_Y\omega = \alpha, \qquad \alpha(Y) = 0, \qquad L_Y\omega = \omega.
\end{equation} 
Since the Reeb vector field $Z$ is a section of $\hbox{ Ker}d\alpha|_M = 0$, hence it is proportional 
to $X_H|_M$. $Z$ can be manifested as a flow of $X_H|_M$ after a time reparametrization 
$dt = f(q,p) \ d\tau$ introduced earlier. Thus we obtain 
$$
Z(x) = \frac{dx}{d\tau} = \frac{dx}{dt}\frac{dt}{d\tau} = f(x)X_H(x), \qquad x=(q,p).
$$
\begin{claim}
The Reeb vector field $Z$ is 
\begin{equation}
Z = \frac{X_H}{Y(H)}, \qquad \hbox{ \text{where} } \qquad f(x) = \frac{1}{Y(H)}.
\end{equation}
\end{claim}
{\bf Proof :} By definition we know $\omega(Y, \cdot ) = \alpha$ and $\alpha(Z) = 1$. Thus we obtain
$$
1 = \alpha(Z) = \omega(Y,f(x)X_H) = f(x)\omega(Y,X_H) = f(x)dH(Y) = f(x)Y(H).
$$
$\Box$
The function $H_0 = H-E/Y(H)$ is defined on $M$ as an invariant surface. Then the 
vector field $X_{H_0}|_M$ is equal to the Reeb vector field $Z$.

\subsubsection{Application to Kepler equation}

Consider a special symplectic transformation $(\bm{p,q}) \to (-\bm{q,p})$\footnote{This transfomation appears in Moser's work on regularization of Kepler orbit}. It is easy to check that 
this transformation leaves the symplectic form:
$$
\omega = d\alpha = \sum_{i=1}^{n}dp_i \wedge dq_i = \sum_{i=1}^{n}-dq_i \wedge dp_i = \sum_{i=1}^{n}d(-q_idp_i) = d{\widetilde \alpha}.
$$ 
The associated Liouville vector field is $Y = \sum_{i=1}^{n} q^i \partial_{q^i}$, which satisfies $\omega(Y, \cdot) = {\widetilde \alpha}$.
It is easy to check that for Kepler Hamiltonian $H = \frac{1}{2}|{\bm p}|^2 - \dfrac{\beta}{|{\bm q}|}$,
$$
Y(H) = \sum_{i=1}^{n}q_i\frac{\partial H}{\partial q_i} = \sum_{i=1}^{n} \big( q_i \big)^2 \frac{\beta}{|{\bm q}|^3} = \frac{\beta}{|{\bm q}|}.
$$
Thus on isoenergetic surface we obtain
$$
\frac{H - E}{Y(H)} = \frac{1}{\frac{\beta}{|{\bm q}|}}\big(\frac{1}{2}|{\bm p}|^2 - \frac{\beta}{|{\bm q}|}\big) = (|{\bm p}|^2 -2E)\frac{|{\bm q}|}{\beta} - 1
= H_0.
$$
Consider a smooth function 
\begin{equation}
F = (H_0 + 1)^2/2 = \frac{(|{\bm p}|^2 - 2E)^2}{8\beta^2}|{\bm q}|^2.
\end{equation}
On the fixed energy surface $H=E$, $F$ becomes $F|_{M_E} = \frac{1}{2}$.
 The trajectories of the Hamiltonian flow of $F$ on the isoenergetic surface 
are governed by the reparametrized time $\tau$. The Hamiltonian vector fields of $F$ and $H_0$  coincide on the
level hypersurface $F = 1/2$ or equivalently $H_0 = 0$. One can easily check
$$
X_F= \frac{|\bm{q}|}{\beta}p_i\frac{\partial}{\partial q_i} - \frac{q_i}{|\bm{q}|^2}\frac{\partial}{\partial p_i} 
= \frac{|\bm{q}|}{\beta}\big((p_i\frac{\partial}{\partial q_i} - \frac{\beta q_i}{|\bm{q}|^3}\frac{\partial}{\partial p_i} \big) $$
$$
= p_i\frac{\partial}{\partial q_i} - \frac{\beta q_i}{|\bm{q}|^3}\frac{\partial}{\partial p_i} / \frac{\beta}{|\bm{q}|} = X_H/Y(H).
$$
Thus we establish regularization theorem due to Moser.
\begin{theorem}
On the isoenergetic surface $F = 1/2$ the trajectories of the
Hamiltonian flow of the function $F = \frac{(|{\bm p}|^2 - 2E)^2}{8\beta^2}|{\bm q}|^2$ traversed
in time $\tau$ equal to trajectories of the Hamiltonian flow of the function $H = \frac{1}{2}|{\bm p}|^2 - \frac{\beta}{|{\bm q}|}$
traversed in real time $t$, and these two times are connected by
$$ \frac{d\tau}{dt} = \frac{\beta}{|{\bm q}|}. $$
\end{theorem}

\subsection{Houri's canonical transformation}

Another canonical transformation that can be applied to the Kepler problem, as performed by Tsuyoshi Houri in \cite{th}, involves swapping the position and momentum phase-space co-ordinates.
\begin{equation}
\widetilde{x}^i = p_i \qquad \widetilde{p}_i = x^i
\end{equation}
Thus, the Kepler hamiltonian will transform as:
\begin{equation}
H = \frac12 \big( p_1^2 + p_2^2 + p_3^2 \big) + \frac \alpha r \quad \longrightarrow \quad \frac12 \Big[ \big( \widetilde{x}^1 \big)^2 + \big( \widetilde{x}^2 \big)^2 + \big( \widetilde{x}^3 \big)^2 \Big] + \frac \alpha{\sqrt{\widetilde{p}_1^2 + \widetilde{p}_2^2 + \widetilde{p}_3^2}}
\end{equation}
As a result, if we choose a fixed energy surface $H = E$ we can further say:
\begin{equation}
\label{houriham} \widetilde{H} = \bigg[ E - \frac{\widetilde{r}^2}2 \bigg]^2 \big( \widetilde{p}_1^2 + \widetilde{p}_2^2 + \widetilde{p}_3^2 \big) = \alpha^2 \hspace{1.5cm} \widetilde{r}^2 = \big( \widetilde{x}^1 \big)^2 + \big( \widetilde{x}^2 \big)^2 + \big( \widetilde{x}^3 \big)^2
\end{equation}
Thus, the related metric with constant curvature $4 E$ on a fixed energy surface is:
\begin{equation}
\widetilde{g}^{ij} (\bm{x}) = \bigg[ E - \frac{\widetilde{r}^2}2 \bigg]^2 \delta^{ij} \hspace{2cm} \widetilde{g}_{ij} (\bm{x}) = \bigg[ E - \frac{\widetilde{r}^2}2 \bigg]^{-2} \delta_{ij}
\end{equation}
So the metric is given by
\begin{equation}
\label{hourimet} ds^2 = \bigg(  E - \frac{|{\bm x}|^2}{2}\bigg)^{-2} d{\bm x}^2 \end{equation}
 If we set the energy to be $E = - \dfrac{k^2}2$, we obtain
\begin{equation}
\label{hourimet1} ds^2 = {4}\big( k + |{\bm x}|^2\big)^{-2} d{\bm x}^2 
\end{equation}
Let $M_k$ be the space of constant curvature manifold. It is known that the Kepler phase space geodesically
incomplete, since in the collision orbits, the particle arrives to the attractive
center with infinite velocity in a finite time, hence does not admit a transitive group of motion.
The mapping of the inversion
$$
I_k : M_k/ \{0\} \to \widehat{M}_k / \{0\},
$$
and $ \bm{x} \to \dfrac{\bm{x}}{|\bm{x}|^2}$ realizes  isometry between its source metric $g$ and the target metric ${\widehat g}$. \\ \\
Suppose 
$$(I_k)_{\ast} : p \mapsto \frac{\bm p}{|\bm{x}|^2} - 2\frac{\bm{x}}{|\bm{x}|^4} \langle \bm{x, p} \rangle, $$ then  one can easily check that
\[ \begin{split}
{\widehat g}_{I(q)} (I_{\ast}x, I_{\ast}x) &= \frac{4}{(1 + k|\bm{x}|^2)^2} \langle I_{\ast}x,I_{\ast}x \rangle \\
&= \frac{4}{\bigg(1 + k\dfrac{1}{|\bm{x}|^2}\bigg)^2} \frac{|\bm{p}|^2}{|\bm{x}|^4} = 4(|\bm{x}|^2 + k)^{-2}|\bm{p}|^2 \ = \ g_q(\bm{x.x}).
\end{split} \]
This describes another conformally flat metric. The question that arises here is; How to 
connect with the Milnor construction? \\ \\
If we set the energy to be $E = - \dfrac{k^2}2$, then we will have
\begin{equation}
\label{houriham} \widetilde{H} = 4 \bigg[ k^2 + \widetilde{r}^2 \bigg]^2 \big( \widetilde{p}_1^2 + \widetilde{p}_2^2 + \widetilde{p}_3^2 \big) = \alpha^2
\end{equation}
If we choose the reparameterization as:
\begin{equation}
\frac{dt}{d \tau} = \frac{r}k
\end{equation}
Then we will have the new Hamiltonian as:
\begin{align}
\mathcal{H} &= \frac rk \bigg( H + \frac{k^2}2 \bigg) = \frac rk \bigg( \frac{|\bm{p}|^2}2 - \frac{\alpha}r + \frac{k^2}2 \bigg) = \frac r{2k} \big( |\bm{p}|^2 + k^2 \big) - \frac{\alpha}k \nonumber \\ \nonumber \\ 
&\qquad \mathbb{H} = k \mathcal{H} + \alpha = \frac r2 \big( |\bm{p}|^2 + k^2 \big)
\end{align}
However, Houri's approach does not preserve the form of equations of the motion or 
geodesic flow operator. That requires another step with Milnor's construction \cite{Milnor}.

\subsection{Milnor's construction}

We shall now separately formulate the Kepler problem under Milnor's construction \cite{Milnor}, which essentially involves a momentum inversion. From this formulation we shall write the metric and the trajectory equation in terms of inverse momentum. \\ \\
The Kepler equation implies:
\begin{equation}
\frac{d \bm{p}}{d t} = - \alpha \frac{\bm{x}}{r^3} \hspace{2cm} \bigg| \frac{d \bm{p}}{d t} \bigg| = \frac{\alpha}{r^2}
\end{equation}
Levi-Civita showed that it is possible to simplify Kepler solutions by introducing a fictitious parameter $\sigma$ such that:
\begin{equation}
\frac{d \sigma}{dt} = \frac1r
\end{equation}
This makes the reparameterized Kepler equation of motion:
\begin{align}
\label{rekep} \frac{d \bm{p}}{d \sigma} = - \alpha \frac{\bm{x}}{r^2} &= \bigg( E - \frac{\big| \bm{p} \big|^2}2 \bigg) \frac{\bm{x}}r \qquad \Rightarrow \qquad \bigg| \frac{d \bm{p}}{d \sigma} \bigg| = \frac{\alpha}r = \frac{\big| \bm{p} \big|^2}2 - E \\ \nonumber \\
\label{milnor} & \therefore \qquad ds^2 = 4 \Big[ 2 E - \big| \bm{p} \big|^2 \Big]^{-2} \big| d \bm{p} \big|^2
\end{align}
Thus, there is one and only one metric on $M_E$ that satisfies our condition. Comparing (\ref{milnor}) result with the Houri's formulation (\ref{hourimet}), we can see that they are identical, except for a swap between momentum and co-ordinate. To describe events in the neighbourhood of infinity, we shall work with the inverted momentum co-ordinate.
\begin{align} \displaybreak[0]
\label{invmom} \bm{w} = \frac{\bm{p}}{\big| \bm{p} \big|^2}, &\qquad \big| \bm{w} \big|^2 = \frac1{\big| \bm{p} \big|^2}, \qquad 2 E \big| \bm{w} \big|^2 < 1 \\ \nonumber \\
\label{diff} \therefore \hspace{0.5cm}  \bm{p} = \frac{\bm{w}}{\big| \bm{w} \big|^2} &\qquad d \bm{p} = \frac{d \bm{w}}{\big| \bm{w} \big|^2} - 2 \frac{\big( \bm{w} . d\bm{w} \big) \bm{w}}{\big| \bm{w} \big|^4} \qquad \big| d \bm{p} \big|^2 = \frac{\big| d \bm{w} \big|^2}{\big| \bm{w} \big|^4}
\end{align}
Using (\ref{rekep}), (\ref{invmom}) and (\ref{diff}) and defining $( \ )' = \frac{d \ }{d \sigma}$, we will get:
\begin{align}
& \qquad \bm{p}' = \bigg( E - \frac{\big| \bm{p} \big|^2}2 \bigg) \frac{\bm{x}}r = \frac{2 E \big| \bm{w} \big|^2 - 1}{2 \big| \bm{w} \big|^2} \frac{\bm{x}}r \nonumber \\ \nonumber \\
\label{traj1} \Rightarrow &\qquad \frac{\bm{w}'}{\big| \bm{w} \big|^2} - 2 \frac{\big( \bm{w} . \bm{w}' \big) \bm{w}}{\big| \bm{w} \big|^4} = \frac{2 E \big| \bm{w} \big|^2 - 1}{2 \big| \bm{w} \big|^2} \frac{\bm{x}}r \\
\label{traj2} \text{and} \qquad \frac{\big| \bm{w}' \big|^2}{\big| \bm{w} \big|^4} &= \bigg( \frac{2 E \big| \bm{w} \big|^2 - 1}{2 \big| \bm{w} \big|^2} \bigg)^2 \qquad \Rightarrow \qquad 4 \big| \bm{w}' \big|^2 = \big( 2 E \big| \bm{w} \big|^2 - 1 \big)^2
\end{align}
If we now substitute the fixed energy level $E = - \dfrac{k^2}2$ in (\ref{traj2}), then we will have the metric in terms of the inverse momentum given as:
\begin{equation}
d s^2 = 4 \Big( 1 + k^2 \big| \bm{w} \big|^2 \Big)^{-2} \big| d \bm{w} \big|^2
\end{equation}
which is the inverse-momentum version of (\ref{milnor}) and a constant mean-curvature metric. From (\ref{traj1}), we get the trajectory equation in terms of inverse momentum as:
\begin{equation}
\bm{x} = \frac{\big| \bm{w} \big|^2 \bm{w}' - 2 \big( \bm{w} . \bm{w}' \big) \bm{w}}{\big| \bm{w} \big|^2 \Big( 2 E \big| \bm{w} \big|^2 - 1 \Big)} r = 2 \alpha \frac{2 \big( \bm{w} . \bm{w}' \big) \bm{w} - \big| \bm{w} \big|^2 \bm{w}'}{\Big( 1 - 2 E \big| \bm{w} \big|^2 \Big)^2}
\end{equation}
Thus, $\bm{x}$ can be epressed as a smooth function of the parameter $\sigma$. If we use $t$ in place 
of $\sigma$, the function stops being smooth only at the point $\bm{x} = 0$. 

\subsection{Geodesic flow}

Now we will see if the form of geodesic flow is preserved after using momentum inversion 
upon Houri's canonical transformation. The Hamiltonian (\ref{houriham}) describing geodesics on 
such spaces under a momentum inversion for $E = - {k}$ \cite{kbs} is given by
\begin{equation}
\widetilde{H} = \frac14 \big( 1 + k \big| \bm{x} \big|^2 \big)^2 \big| \bm{p} \big|^2
\end{equation}
%setting $\beta = \dfrac{\widetilde{H}^{\frac12}}{k^2 x_i}$
From this Hamiltonian, setting we can derive the Hamiltonian flow vector field
\[ \begin{split}
X_{\widetilde{H}} &= \frac{\partial \widetilde{H}}{\partial p_i} \frac{\partial \ }{\partial x^i} - \frac{\partial \widetilde{H}}{\partial x^i} \frac{\partial \ }{\partial p_i} \\
&= \frac{1}{2}\big( 1 + k \big| \bm{x} \big|^2 \big)^2 p^i \frac{\partial \ }{\partial x^i} - 
\big( 1 + k \big| \bm{x} \big|^2 \big) \big| \bm{p} |^2 k x_i \frac{\partial \ }{\partial p_i} \\
&= 2 \widetilde{H}\frac{p^i}{\big| \bm{p} \big|^2} \frac{\partial \ }{\partial x^i} - 2k\widetilde{H}^{\frac12}| \bm{p} \big| x_i
\frac{\partial \ }{\partial p_i} \hspace{0.5cm} = \quad 2\widetilde{H}^{\frac12} \big| \bm{p} \big| \bigg[\frac{p^i}{\big| \bm{p} \big|^3} \frac{\partial \ }{\partial x^i} - kx_i \frac{\partial \ }{\partial p_i} \bigg]
\end{split} \]
Thus we finally obtain
\begin{equation}
\therefore \qquad \big( 2\widetilde{H}^{\frac12} \big| \bm{p} \big| \big)^{-1} X_{\widetilde{H}} = \frac{p^i}{\big| \bm{p} \big|^3} \frac{\partial \ }{\partial x^i} - x_i \frac{\partial \ }{\partial p_i}
\end{equation}
Comparing the flow operator above with the geodesic flow in (\ref{kepflow}), we obtain the quasi-
Hamiltonian vector field
of Kepler equation in momentum space
\begin{equation}
X_{Kepler}^{mom} = \big( 2 k^2 \widetilde{H}^{\frac12} \big| \bm{p} \big| \big)^{-1} X_{\widetilde{H}}
\end{equation}
Thus, we can see that combining Houri's transformation with Milnor's momentum inversion preserves the form of the geodesic flow, aside from a momentum factor.
%&= 2 \widetilde{H} \big| \bm{p} \big| \bigg[ \frac{p^i}{\big| \bm{p} \big|^3} \frac{\partial \ }{\partial x^i} - \frac{k^2 x_i}{\widetilde{H}^{\frac12}} \frac{\partial \ }{\partial p_i} \bigg] \\
\smallskip
A vector field ${X}$ on a symplectic manifold $(M, \omega)$ is quasi-Hamiltonian if
there exists a (nowhere-vanishing) function $\Lambda$ such that ${X}$ is a Hamiltonian vector field
$\Lambda X \in {\cal X}_H(M )$, thus $i_{\Lambda X} \omega = dH$.
This condition can alternatively be written as as $i_X (\Lambda \omega) = dH$, but the point is that
the 2-form $\Lambda \omega$ is not closed in the general case. \\
\smallskip
By applying the special canonical transformation that interchanges ${\bm x}$ and ${\bm p}$, the Kepler equation on momentum space 
transforms to the usual Kepler equation with the Hamiltonian
$$
H = \frac{1}{4} \big( k + |{\bm p}|^2 \big)^2 |{\bm x}|^2.
$$
Finally, we will explore the results of parameterizing the JM metric and Kepler equation with the eccentric anomaly.

\subsection{JM metric and Kepler equation parametrized by eccentric anomaly}

The Kepler Hamiltonian is given as
\begin{equation}
H = \frac12 \sum_{n=1}^3 \big( p_n \big)^2 - \frac \alpha r \hspace{2cm} r^2 = \sum_{n=1}^3 \big( x^n \big)^2
\end{equation}
Let us perform the following canonical transformation:
\begin{equation}
x^i \longleftrightarrow p_i : \qquad \big( x^i, p_i \big) = \big( \widetilde{p}_i, \widetilde{x}^i \big)
\end{equation}
Setting $H = E$, this transformation allows us to write a new Hamiltonian $\widetilde{H}$ as
\[ \begin{split}
E = \frac12 \sum_{n=1}^3 \big( \widetilde{x}^n \big)^2 - \frac \alpha{\sqrt{\sum_{n=1}^3 \widetilde{p}_n^2}} \qquad &\Rightarrow \qquad E - \frac12 \sum_{n=1}^3 \big( \widetilde{x}^n \big)^2 = - \frac \alpha{\sqrt{\sum_{n=1}^3 \widetilde{p}_n^2}} \\
\Rightarrow \qquad \widetilde{H} = \bigg[ E - \frac12 \sum_{n=1}^3 &\big( \widetilde{x}^n \big)^2 \bigg]^2 \sum_{n=1}^3 \widetilde{p}_n^2 = \alpha^2
\end{split} \]
The Hamilton’s equations for this canonically transformed system (\ref{houriham}) are:
\begin{equation} \label{hameq}
\begin{split}
\dot{\widetilde{x}}^i &= \frac{\partial \widetilde{H}}{\partial \widetilde{p}_i} = 2 \bigg[ E - \frac12 \sum_{n=1}^3 \big( \widetilde{x}^n \big)^2 \bigg]^2 \widetilde{p}_i \\
\dot{\widetilde{p}}_i &= \frac{\partial \widetilde{H}}{\partial \widetilde{x}^i} = 2 \bigg[ E - \frac12 \sum_{n=1}^3 \big( \widetilde{x}^n \big)^2 \bigg] \bigg( \sum_{n=1}^3 \widetilde{p}_n^2 \bigg) \widetilde{x}^i \\
\end{split}
\end{equation}
To proceed to equations of motion, we shall use (\ref{hameq}) to write:
\[ \begin{split}
\ddot{\widetilde{x}}^i &= 2 \frac{d \ }{d t} \bigg[ \bigg\{ E - \frac12 \sum_{n=1}^3 \big( \widetilde{x}^n \big)^2 \bigg\}^2 \widetilde{p}_i \bigg] \\
&= - 4 \bigg\{ E - \frac12 \sum_{n=1}^3 \big( \widetilde{x}^n \big)^2 \bigg\} \bigg( \sum_{k = 1}^3 \widetilde{x}^k \dot{\widetilde{x}}^k \bigg) \widetilde{p}_i + 2 \bigg\{ E - \frac12 \sum_{n=1}^3 \big( \widetilde{x}^n \big)^2 \bigg\}^2 \dot{\widetilde{p}}_i \\
&= - 2\bigg\{ E - \frac12 \sum_{n=1}^3\big( \widetilde{x}^n \big)^2 \bigg\}^{-1} \widetilde{x}^k \dot{\widetilde{x}}^k \bigg) \widetilde{x}^i 
+ 4 \bigg\{ E - \frac12 \sum_{n=1}^3 \big( \widetilde{x}^n \big)^2 \bigg\} \widetilde{H} \widetilde{x}^i \\
& =- 2\frac{\big(\widetilde{\bm x} \cdot \dot{\widetilde{\bm x}}\big)\cdot \dot{\widetilde{\bm x}}}{\Lambda} + 4\Lambda {\widetilde H}{\widetilde{\bm x}}
\end{split} \]
where $\Lambda =  (E - \frac12 \sum_{n=1}^3 \big( \widetilde{x}^n \big)^2)$.
Let us write $\widetilde{\bm x}$ as ${\bm x}$, hence we obtain
 
%&= - 8 \bigg\{ E - \frac12 \sum_{n=1}^3 \big( \widetilde{x}^n \big)^2 \bigg\}^3 \bigg( \sum_{k = 1}^3 \widetilde{x}^k \widetilde{p}_k \bigg) \widetilde{p}_i + 4 \bigg\{ E - \frac12 \sum_{n=1}^3 \big( \widetilde{x}^n \big)^2 \bigg\}^3 \bigg( \sum_{n=1}^3 \widetilde{p}_n^2 \bigg) \widetilde{x}^i \\
%&= 4 \bigg\{ E - \frac12 \sum_{n=1}^3 \big( \widetilde{x}^n \big)^2 \bigg\}^3 \bigg[ \bigg( \sum_{n=1}^3 \widetilde{p}_n^2 \bigg) \widetilde{x}^i - 2 \bigg( \sum_{k = 1}^3 \widetilde{x}^k \widetilde{p}_k \bigg) \widetilde{p}_i \bigg]
%\end{split} \]
\begin{equation}
\label{equation} \ddot{{\bm x}} = -2\frac{\big({\bm x} \cdot \dot{{\bm x}}\big)\cdot \dot{{\bm x}}}{\Lambda} + 
4\Lambda {\widetilde H}{{\bm x}}.
\end{equation}
It is known that  the Laplace Lenz Runge vector
\begin{equation}\label{LRL}
A({\bm x}, \dot{\bm x})) = \frac{1}{\mu}(2H + \frac{\mu}{|{\bm x}|}){\bm x} - \frac{1}{\mu}({\bm x}\cdot \dot{\bm x})\dot{\bm x}
\end{equation}
is a conserved quantity for the Kepler flow, we can re-write this equation using $A({\bm x}, \dot{\bm x})$. 
Using the Laplace Lenz Runge vector we obtain
$$
\frac{2\mu}{\Lambda}\big( A({\bm x}, \dot{\bm x}) - \frac{\bm x}{|{\bm x}|}\big) = 
-2\frac{\big({\bm x} \cdot \dot{{\bm x}}\big)\cdot \dot{{\bm x}}}{\Lambda} + 
4\Lambda {\widetilde H}{{\bm x}},
$$
where $E = {\widetilde H}\Lambda^2$. Thus equation (\ref{equation}) can be written as
\begin{equation}
 \qquad \ddot{{\bm x}} + \frac{2\mu}{\Lambda}\frac{\bm x}{|{\bm x}|} =  \frac{2\mu A}{\Lambda}.
\end{equation}

\bigskip

\subsubsection{Kepler equation parameterizing the eccentric anomaly}

An advantage of the eccentric anomaly is that it is well suited to describe Kepler motion 
in position space. Therefore we derive the equation of motion w.r.t. this parameter. \\ \\
Let us reparametrize the time as
\begin{equation}
dt = \frac{|{\bm x}|}{\epsilon} ds.
\end{equation}
Thus we obtain 
$$
\frac{d{\bm x}}{ds} = \frac{d{\bm x}}{dt}\frac{dt}{ds} = \dot{\bm x} \frac{|{\bm x}|}{\epsilon}.
$$
The second derivative yields
\[ \frac{d^2{\bm x}}{ds^2} = \frac{1}{|{\bm x}|^2} \bigg({\bm x}\cdot \frac{d{\bm x}}{ds} \bigg) \frac{d{\bm x}}{ds}
 + \frac{|{\bm x}|^2}{\epsilon^2}\ddot{\bm x} =\frac{1}{|{\bm x}|^2} \bigg({\bm x}\cdot \frac{d{\bm x}}{ds} \bigg) \frac{d{\bm x}}{ds}
- \frac{\mu}{\epsilon^2|{\bm x}|} {\bm x} \]
where we have used the Kepler equation $\ddot{\bm x} = - \dfrac{\mu {\bm x}}{|{\bm x}|^3}$.
The Laplace Lenz Runge vector
\[ \begin{split}
A({\bm x}, \dot{\bm x})) &= \frac{1}{\mu} \bigg(2E + \frac{\mu}{|{\bm x}|}){\bm x} - \frac{1}{\mu}({\bm x}\cdot \dot{\bm x} \bigg)\dot{\bm x} \\
&= -\frac{\epsilon^2}{\mu}\bigg[ \frac{1}{|{\bm x}|^2} \bigg({\bm x}\cdot \frac{d{\bm x}}{ds} \bigg) \frac{d{\bm x}}{ds} - 
\frac{\mu}{\epsilon^2}\frac{{\bm x}}{|{\bm x}|} + {\bm x} \bigg] \qquad \qquad \epsilon^2 = -2E.
\end{split} \]
Here we consider the case of negative energy, ie. bounded orbits.
Therefore we obtain
\begin{equation}
\frac{d^2{\bm x}}{ds^2} + {\bm x} = -\frac{\mu}{\epsilon^2} \bm{A}.
\end{equation}
Let us start with the Hamiltonian
$$
 \widetilde{H} = \alpha^2 = \bigg[ E - \frac{\widetilde{r}^2}2 \bigg]^2 \bigg( \sum_{n=1}^3 \widetilde{p}_n^2 \bigg)
= \frac{1}{4}\bigg[\epsilon^2 + |{\bm x}|^2\bigg]^2 |{\bm p}|^2,
$$
where we have used $2E = -\epsilon^2$. \\ \\
Define 
$$ G({\bm x},{\bm p}) =  \widetilde{H}^{1/2} =  \frac{1}{2}\big(\epsilon^2 + |{\bm x}|^2\big) |{\bm p}|.$$
We now consider regularized Kepler Hamiltonian system. The system of the Hamiltonian 
obtained from
\begin{equation}
 \widetilde{G}({\bm x},{\bm p}) =   \frac{1}{2\epsilon}\big(\epsilon^2 + |{\bm x}|^2\big) |{\bm p}| - 
\frac{\mu}{\epsilon}, \qquad \epsilon \neq 0,
\end{equation}
is given by
$$
\dot{\bm p} = |{\bm p}|{\bm x} \qquad \dot{\bm x} = -\frac{1}{|{\bm p}|^2} \big(\widetilde{G}({\bm x},{\bm p}) + 
\frac{\mu}{\epsilon}\big)
$$
By the first equation, ${\bm x} = \dfrac{\epsilon}{|{\bm p}|}\dot{\bm p}$, we obtain
$$
\ddot{\bm p} = \frac{1}{|{\bm p}|^2}({\bm p}\cdot \dot{\bm p})\dot{\bm p} - \frac{1}{\epsilon^2|{\bm p}|}
\big(\widetilde{G}({\bm x},{\bm p}) + \mu \big){\bm p}.
$$
Its restriction to the level set $\big[({\bm x},{\bm p}) | \widetilde{G}({\bm x},{\bm p}) = 0 \big]$ is flow of the Kepler problem in the 
momentum space parametrized by the eccentric anomaly.

\vspace{-0.1cm}
\section{Conclusion}

So far, we have seen that the Jacobi metric transforms the dynamics from both time-independent Lagrangian and Hamiltonian perspective from a space with potential functions to an equivalent free particle geodesic of lower dimension. All aspects of integrability and the first integrals are preserved under such lifts. The Hamiltonian and Lagrangian of such metrics possess a conformal factor and equate to unity. Such a procedure can cast the TeVeS theory into the form of a Kaluza-Klein construction \cite{lgs}.

When applied to the Kepler problem, this holds true. Such a transformation for a particular energy level combined with Bohlin's canonical transformation converts the isotropic oscillator problem to the Kepler problem. Houri's canonical transformation is found to be incomplete without Milnor's momentum inversion map, which preserves the form of geodesic flows as identical to that of the Kepler problem. \\

There are quite a few areas of Jacobi-Maupertuis have been less studied, for example, 
the Maupertuis principle can be used in the construction of the theory of many-valued 
functionals, which arises naturally in the study of the motion of charged particle in a 
scalar potential field in the presence of magnetic field \cite{Novikov}. It would be interesting to 
extended this project to the study of integrable magnetic geodesic flows \cite{iat1, iat2}. Recently 
this has been extended in \cite{BKF} to present a modern outlook to describe the mechanism of 
the Maupertuis principle using classical integrable dynamical systems. This mechanism 
yields integrable geodesic flows and integrable system associated to curved spaces. In 
fact other related topics like the formulation of the Jacobi metric for time-like geodesics 
and its application to curved space-time \cite{Gibbons}, applications of geodesic instabilities for the 
planar gravitational three-body problem \cite{KrishSena} should get more attention. The application 
of this analysis to the generalized MICZ-Kepler problem would be fascinating.

\vspace{-0.1cm}
\section*{ Acknowledgement  } 

PG would like to thank Marco Cariglia for his correspondence and many invauable comments. We are
also grateful to Tudor Ratiu, Pepin Cari\~nena, Guowu Meng and E. Harikumar for their
correspondences, collaboration and contributions to this project.

\end{document}